\documentclass[12pt,a4paper,reqno]{amsart}
\usepackage{graphicx}
\usepackage{amssymb}
\usepackage{cite}
\usepackage{amsmath}
\usepackage{latexsym}
\usepackage{amscd}
\usepackage{tikz}
\usepackage{amsthm}
\usepackage{mathrsfs}
\usepackage{url}
\usepackage[utf8]{inputenc}
\usepackage[english]{babel}
\usepackage{amsfonts}
\usepackage{mathtools}
\usepackage[margin=1.2in]{geometry}
\usepackage{natbib}
\usepackage{algorithm}
\usepackage{algpseudocode}
\usepackage{listings}
\usepackage{xcolor}
\usepackage{caption}
\usepackage{subcaption}
\usepackage{bbm}
\emergencystretch=\maxdimen
\usepackage{listings}
\usepackage{xcolor}
\usepackage{natbib}
\definecolor{codegreen}{rgb}{0,0.6,0}
\definecolor{codegray}{rgb}{0.5,0.5,0.5}
\definecolor{codepurple}{rgb}{0.58,0,0.82}
\definecolor{backcolour}{rgb}{0.95,0.95,0.92}

\lstdefinestyle{mystyle}{
	backgroundcolor=\color{backcolour},   
	commentstyle=\color{codegreen},
	keywordstyle=\color{magenta},
	numberstyle=\tiny\color{codegray},
	stringstyle=\color{codepurple},
	basicstyle=\ttfamily\footnotesize,
	breakatwhitespace=false,         
	breaklines=true,                 
	captionpos=b,                    
	keepspaces=true,                 
	numbers=left,                    
	numbersep=5pt,                  
	showspaces=false,                
	showstringspaces=false,
	showtabs=false,                  
	tabsize=2
}

\usepackage{hyperref}
\hypersetup{
	colorlinks=true,
	linkcolor=blue,
	filecolor=magenta,      
	urlcolor=cyan,
}

\definecolor{shadecolor}{rgb}{1,0.8,0}

\DeclareMathOperator*{\argmax}{arg\,max}

\usepackage{setspace} 

\numberwithin{equation}{section}
\DeclarePairedDelimiter\floor{\lfloor}{\rfloor}

\newtheorem{theorem}{Theorem}[section]

\newtheorem{proposition}[theorem]{Proposition}

\newtheorem*{Main Result}{Main Result}

\theoremstyle{remark}

\theoremstyle{definition}
\newtheorem{definition}[theorem]{Definition}

\theoremstyle{theorem}

\theoremstyle{theorem }

\theoremstyle{remark}

\newcommand*{\emaill}[1]{\texttt{#1}}

\numberwithin{equation}{section}

\title{Online Revenue Maximization with Unknown Concave Utilities}

\begin{document}
	
	\author[Owen Shen]{
		Owen Shen
		\\
		\vspace{0.1cm}
		\\
		Department of Mathematics
		\vspace{0.1cm}
		\\
		\MakeLowercase{
			\emaill{owenshen@stanford.edu}}\\
		\vspace{0.2cm}
		Stanford University}
	\begin{abstract}
	 We study an online revenue maximization problem where the consumers arrive i.i.d from some unknown distribution and purchase a bundle of products from the sellers. The classical approach generally assumes complete knowledge of the consumer utility functions, while recent works have been devoted to unknown linear utility functions. This paper focuses on the online posted-price model with unknown consumer distribution and unknown consumer utilities, given they are concave.  Hence, the two questions to ask are i) when is the seller's online maximization problem concave, and ii) how to find the optimal pricing strategy for non-linear utilities. We answer the first question by imposing a third-order smoothness condition on the utilities. The second question is addressed by two algorithms, which we prove to exhibit the sub-linear regrets of $O(T^{\frac{2}{3}} (\log T)^{\frac{1}{3}})$ and $
	O(T^{\frac{1}{2}} (\log T)^{\frac{1}{2}})$ respectively. 
	\end{abstract}
	
	\maketitle
	
	\tableofcontents
	\section{Introduction}
	The online posted-price model plays a central role in the online maximization problem \cite{21}. In this model, each consumer arrives in the market and purchases a bundle of products from the seller according to their strictly convex utilities, and the seller is allowed to adjust its price and observe the consumer reactions. In our model, we assume no explicit budge constraint like in \cite{20}, which can be formulated into the convex utility itself. The type of seller we focus on is a third-party seller, namely that the seller is not constrained by its resource capacity but pays a cost for each product. For example, a retailer that imports goods from a large manufacturing firm. So unlike online maximization problems in \cite{24} constrained by its resource efficiency, our main constraint is the product costs. The majority of existing approaches to solving problems of this kind in the literature assume complete knowledge of the consumer preference (\cite{22}, \cite{23}), an assumption not necessarily true in the application. Online market research and consumer surveys are the necessary components of the online maximization problem, and it is not explicitly unknown to us how sensitive the optimal value is given the approximation errors generated from those surveys. This paper, therefore, aims to fill the gap of online learning in the context of unknown utilities. 

Related contributions that also focus on the online maximization problem with partially unknown objective function include \cite{20} that analyzes online Fisher market with linear utilities; \cite{15} that analyzes online linear programming with unknown demand parameter; \cite{16} that analyzes online matching problem with unknown coupling rewards; and \cite{17} that analyzes online resource allocation with unknown consumer preferences for new products. The most related result is \cite{28}, where a maximization problem on a single-item, unknown-utility model is discussed. Their regret is defined by time, where the consumer arrives at a random speed; whereas our regret is defined by the total number of consumers. 

We assume in our model $n$ discrete types of consumers each type with a fixed but different utility, and they arrive sequentially in i.i.d fashion to the seller from some unknown distribution $\mathcal{P}$. This type of i.i.d stochastic is widely used in the literature like \cite{13} and \cite{14}, while a permutation model like \cite{11} and \cite{12} is also popular. There is also a rising research area to analyze more general stochastic inputs like \cite{24}.

The consumer's purchased quantities and types are revealed to the seller, belonging to a class of revealed preference and posted-price problem \cite{25}. The main difficulty is that, only the point-wise utilities are revealed. So the seller is generally unknown of the shape of consumer utilities even after the conclusion of the selling. 

Therefore, the seller with $m$ products must approximate the consumer distributions and, given its approximated distribution, find an optimal pricing strategy. The first question, answered in section 2, is to ask when is the seller's problem concave. If the problem is concave, finding an optimal pricing strategy would be relatively easy. For example, one can employ a binary search for the optimality value. The second question is to discuss how would the errors in the distributional approximation and the errors in the optimal-solution approximation affect the optimal value. In section three, the former is characterized in terms of the convergence rate of the total variation distance of the empirical distribution, and the latter is characterized in terms of a linear function of the Euclidean distance of the approximated optimal solution to the true solution. Finally, in sections four and five, we hope to construct algorithms to handle this problem. The two classical algorithms are a one-time learning algorithm and a dynamical (geometrically learning) algorithm. Other applications of the one-time learning algorithm can be found in \cite{19} and the dynamical learning algorithm in \cite{14}. 

In conclusion, our results establish the significance of classical algorithms in a new context that requires the minimum information from the online market. This generality is the main contribution of this paper.

	\section{When Is The Online Maximization A Concave Problem}
We start with an online market model with $n$ types of customers and $m$ types of products. Each type of customer $i$ has a strictly concave utility function with quantity parameter denoted as $k$: $$v_i(j,k):\{1,\cdots, m\}\times \mathbb{R}_+\to \mathbb{R}_+$$
such that the total utility type $i$ customer receives after purchasing quantities $k_j$ of each type of product $j$ is
$$\sum_{j=1}^{m}v_i(j,k_j). $$
For any rational buyer, under the assumption of no budget constraint (the budget constraint can be formulated into the concave utility functions), she will purchase enough product of type $j$ until the marginal increase of the utility matches with the price of that product. Hence, her optimal purchased order, with a given price $p\in \mathbb{R}^m$ is
$$ k^*=\argmax_{k\in \mathbb{R}^m} \sum_{j=1}^{m}v_i(j,k_j)-k_j\cdot p.$$
When the problem is feasible and since each $v_j$ is a concave, differentiable and linearly separable utility, we have 
$$k^*_j=\left[\frac{\partial }{\partial k_j}v_i(j,k_j)\right]^{-1}(p_j).$$
The revenue collected by selling at price $p$ is therefore to customer type $i$ is 
\begin{equation}\sum_{j=1}^{m}\left[\frac{\partial }{\partial k_j}v_i(j,k_j)\right]^{-1}(p_j)\cdot (p_j-c_j).\label{eq:1}\end{equation}
Let $i$ be drawn according to the distribution of the types of the customer, we now formulate the online maximization problem into the following: 

$$\max_{p\in \mathbb{R}^m}\sum_{j=1}^{m}\mathbb{E}_{i\sim \mathcal{P}}\left\{\left[\frac{\partial }{\partial k_j}v_i(j,k_j)\right]^{-1}(p_j)\cdot p_j\right\}.$$
The goal is to ask when is the above revenue a convex problem in terms of $p$. By our linear separability assumption and the fact that conic combination preserves convexity, it suffices to check the convexity and maximization condition over a single dimension $j$ for a single product 
\begin{equation}\max_{p_j}\left[\frac{\partial }{\partial k_j}v_i(j,k_j)\right]^{-1}(p_j)\cdot (p_j-c_j).\label{eq:2}\end{equation}
To find the conditions under which the above problem is concave, we need to check the following is negative:
\begin{equation} \frac{\partial^2 }{\partial k_j^2} \left[\frac{\partial }{\partial k_j}v_i(j,k_j)\right]^{-1}(p_j)\cdot p_j
=\frac{2-(p_j-c_j)\frac{\partial^3}{\partial k_j^3} v_i(j,\cdot)\circ \left[ \frac{\partial}{\partial k_j} v(\cdot)\right]^{-1}(p_j)}{\frac{\partial^2}{\partial k_j^2} v_i(j,\cdot)\circ \left[ \frac{\partial}{\partial k_j} v(\cdot)\right]^{-1}(p_j)}.
\end{equation}
Since we know $v_j$ is concave, it suffices to check that 
\begin{equation}2-(p_j-c_j)\frac{\partial^3}{\partial k_j^3} v_i(j,\cdot)\circ \left[ \frac{\partial}{\partial k_j} v(\cdot)\right]^{-1}(p_j)\geq 0.\end{equation}
If we standardize the price as $p_j-c_j\in [0,1]$, we can require a smoothness condition 
$$\frac{\partial^3}{\partial k_j^3}v_i(j,k_j)\leq 2
$$
to ensure that our problem is concave. In other words, we require the curvature of the utility function not to increase too much. So if we summarize our above derivation, we get the following theorem:

\begin{theorem}\label{theorem 1}
Let $v$ be a linearly separable utility function with dimension $\mathbb{R}^n\times \mathbb{R}^m$
such that for each $v_i(j,k_j)$, we have 
$$\frac{\partial^3}{\partial k_j^3}v_i(j,k_j)\leq 2.
$$
Furthermore, the price parameter $p-c\in \mathbb{R}^m\in [0,1]^m$ is non-discriminating. Then, the optimization problem $$\max_{p\in \mathbb{R}^m}\sum_{j=1}^{m}\mathbb{E}_{i\sim \mathcal{P}}\left\{\left[\frac{\partial }{\partial k_j}v_i(j,k_j)\right]^{-1}(p_j)\cdot (p_j-c_j)\right\}$$
 is a concave problem in $p$. 
\end{theorem}
Since the problem is now concave, we can perform fast search algorithms to find the optimal pricing strategy. In particular, we hope to analyze the optimization problem in online fashion with noise from the market.

There are two things needed to be learned in our problem: the underlying distribution of the consumer and the consumer utilities. We assume at each time a consumer reveals her types by her purchasing behaviors, so branching is not allowed. The main issue is to learn the optimal pricing strategy. The method we will consider here is a simple binary search for the maximum point of a concave function. Any search algorithm with an exponential convergence rate will produce the same result using our analysis. Since we will soon see that the empirical distribution can only converge at an exponential rate, any search algorithm better than the exponential rate will not improve the order of the regret. We will consider two algorithms: the one-time learning algorithm and the geometric learning algorithm.

\section{Regrets Of Algorithms}
Before we analyze the efficiency of the algorithms, we first discuss our notion of regret:
\begin{definition}
We define the regret of our algorithms with a pricing policy $p(t)$ at time $t$, denoted as $\pi$, of operation horizon $T$ as 
 $$R_T(\pi):=T\max_{p\in \mathbb{R}^m}\sum_{j=1}^{m}\mathbb{E}_{i\sim \mathcal{P}}\left\{\left[\frac{\partial }{\partial k_j}v_i(j,k_j)\right]^{-1}(p_j)\cdot (p_j-c_j)\right\}$$$$-\sum_{j=1}^{m}\mathbb{E}_{i\sim \mathcal{P}}\sum_{t=1}^{T}\left\{\left[\frac{\partial }{\partial k_j}v_i(j,k_j)\right]^{-1}(p_j(t))\cdot (p_j(t)-c_j)\right\}$$
 where $p_j(t)$ is the price of product $j$ the algorithm provides at time
 $t$.\end{definition} 
 This notion of regret is different from another popular choice known as the competitive-ratio like in \cite{17}, for we are interested in the long-term performance of the algorithms, instead of the fixed-time competitiveness. 
 
 Now we are interested in knowing how would the approximation errors affect optimal value. We can use first order analysis to write, using $p_j^*$ to denote the optimal pricing,
 
 \begin{align}
  &R_T(\pi)\\
  \leq &\sum_{j=1}^{m} \mathbb{E}_{i\sim \mathcal{P}} \sum_{t=1}^{T} |p_j(t)-p_j^*|\max_{p_j}\frac{\partial }{\partial p_j}\left\{\left[\frac{\partial }{\partial k_j}v_i(j,k_j)\right]^{-1}(p_j)\cdot (p_j-c_j)\right\}\\
  =&  \sum_{j=1}^{m} \mathbb{E}_{i\sim \mathcal{P}} \sum_{t=1}^{T} |p_j(t)-p_j^*|\max_{p_j} \left\{\left[\frac{\partial }{\partial k_j}v_i(j,k_j)\right]^{-1}(p_j)+\frac{(p_j-c_j)}{\frac{\partial^2}{\partial k_j^2} v_i(j,\cdot)\circ \left[ \frac{\partial}{\partial k_j} v(\cdot)\right]^{-1}(p_j)}\right\}
 \end{align}
 Let us observe that 
$$\left[\frac{\partial }{\partial k_j}v_i(j,k_j)\right]^{-1}(p_j)$$ is the quantity purchased. Since the utility is strictly convex, we know that at the lowest possible price $p_j=c_j>0$, the consumer purchase a finite number of products. Moreover, suppose that for some positive $\epsilon$, we have for all feasible $p_j$,
$$\frac{\partial^2}{\partial k_j^2} v_i(j,\cdot)\circ \left[ \frac{\partial}{\partial k_j} v(\cdot)\right]^{-1}(p_j)<-\epsilon.$$
Then both quantities are finite. This establishes that the difference between the optimal and the sub-optimal Lipschitz in $p$. 

To summarize our observation, we have the following proposition:
\begin{proposition}\label{regret}
Let the conditions of Theorem \ref{theorem 1} be satisfied and for some positive $\epsilon$, we have for all feasible $p_j$,
$$\frac{\partial^2}{\partial k_j^2} v_i(j,\cdot)\circ \left[ \frac{\partial}{\partial k_j} v(\cdot)\right]^{-1}(p_j)<-\epsilon.$$
Then, there exists some finite constant $M$ such that 
$$R_T(\pi)\leq M\sum_{t=1}^{T}\|p(t)-p^*\|.$$
\end{proposition}
Moreover, we are not only interested in the sensitivity of the optimal value regarding the pricing strategy, we are also interested in its sensitivity over the underlying distributions, for our algorithms have two sources of errors: the approximations in the distributions and the approximation in the optimal pricing given the approximated distributions. Our algorithms use empirical distribution $	\hat{P}_t^{(*)}$ to solve for the optimal pricing, so we are interested in the difference 

 \begin{align}
 &\max_{p\in \mathbb{R}^m}\sum_{j=1}^{m}\mathbb{E}_{i\sim \mathcal{P}}\left\{\left[\frac{\partial }{\partial k_j}v_i(j,k_j)\right]^{-1}(p_j)\cdot (p_j-c_j)\right\}\\
 -&\max_{p\in \mathbb{R}^m}\sum_{j=1}^{m}\mathbb{E}_{i\sim \hat{P}_t^{(*)}}\left\{\left[\frac{\partial }{\partial k_j}v_i(j,k_j)\right]^{-1}(p_j)\cdot (p_j-c_j)\right\}.
 \end{align}
If we denote the empirical distribution of type $i$ consumer as $\hat{w}_i$ while the true distribution is $w_i$, we have the following bound using total variation distance:

 \begin{align}
 &\max_{p\in \mathbb{R}^m}\sum_{j=1}^{m}\mathbb{E}_{i\sim \mathcal{P}}\left\{\left[\frac{\partial }{\partial k_j}v_i(j,k_j)\right]^{-1}(p_j)\cdot (p_j-c_j)\right\}\\
 \leq & \max_{p\in \mathbb{R}^m}\sum_{j=1}^{m}\sum_{i=1}^{n}\min\{\hat{w}_i,w_i \}\left\{\left[\frac{\partial }{\partial k_j}v_i(j,k_j)\right]^{-1}(p_j)\cdot (p_j-c_j)\right\}\\
 +& \max_{p\in \mathbb{R}^m}\sum_{j=1}^{m}\sum_{i=1}^{n}\left(w_i-\min\{\hat{w}_i,w_i \}\right)\left\{\left[\frac{\partial }{\partial k_j}v_i(j,k_j)\right]^{-1}(p_j)\cdot (p_j-c_j)\right\}\\
 \leq &  \max_{p\in \mathbb{R}^m}\sum_{j=1}^{m}\sum_{i=1}^{n}\min\{\hat{w}_i,w_i \}\left\{\left[\frac{\partial }{\partial k_j}v_i(j,k_j)\right]^{-1}(p_j)\cdot (p_j-c_j)\right\}\\
 +& d_{TV}(\mathcal{P}, \hat{P}_t^{(*)})\max_{p\in \mathbb{R}^m} \sum_{j=1}^{m}\left[\frac{\partial }{\partial k_j}v_i(j,k_j)\right]^{-1}(p_j)\\
 \leq &  \max_{p\in \mathbb{R}^m}\sum_{j=1}^{m}\sum_{i=1}^{n}\hat{w}_i\left\{\left[\frac{\partial }{\partial k_j}v_i(j,k_j)\right]^{-1}(p_j)\cdot (p_j-c_j)\right\} +d_{TV}(\mathcal{P}, \hat{P}_t^{(*)}) M'\\
 \leq & \max_{p\in \mathbb{R}^m}\sum_{j=1}^{m}\mathbb{E}_{i\sim \hat{P}_t^{(*)}}\left\{\left[\frac{\partial }{\partial k_j}v_i(j,k_j)\right]^{-1}(p_j)\cdot (p_j-c_j)\right\}+M'd_{TV}(\mathcal{P}, \hat{P}_t^{(*)})
 \end{align}
 where we have 
 $$\max_{p\in \mathbb{R}^m} \sum_{j=1}^{m}\left[\frac{\partial }{\partial k_j}v_i(j,k_j)\right]^{-1}(p_j)\leq M'$$
for some constant $M'$ because the utility is concave such that the maximum purchase quantity is finite given the nonzero cost of each product. So, to summarize our observation, we have the following regret:

\begin{theorem}
\label{regret theorem}
Let the conditions of Theorem \ref{theorem 1} be satisfied and for some positive $\epsilon$, we have for all feasible $p_j$,
$$\frac{\partial^2}{\partial k_j^2} v_i(j,\cdot)\circ \left[ \frac{\partial}{\partial k_j} v(\cdot)\right]^{-1}(p_j)<-\epsilon.$$
Then, there exists some finite constant $M$ and $M'$ such that 
$$R_T(\pi)\leq \sum_{t=1}^{T}\left(M\|p(t)-p(t)^*\|+M'd_{TV}(\mathcal{P}, \hat{P}_t^{(*)})\right),$$
where $\hat{P}_t^{(*)}$ is the latest available approximation of the distribution and $$p(t)^*:=\argmax_{p\in \mathbb{R}^m} \sum_{j=1}^{m}\mathbb{E}_{i\sim \hat{P}_t^{(*)}}\left\{\left[\frac{\partial }{\partial k_j}v_i(j,k_j)\right]^{-1}(p_j)\cdot (p_j-c_j)\right\}. $$
\end{theorem}
Hence, this theorem establishes that if we can approximate the underlying distribution of the consumer and the optimal pricing given the approximated distribution fast enough, we get a small regret. In the next section, we will particularize how to approximate the underlying distributions and solve for optimal pricing.  

Finally, to end this section, let us discuss the conditions listed in Theorem \ref{regret theorem}, which may seem to be restrictive. Actually, many standard utility functions satisfy the given regularity condition.
\subsubsection{Log-like Utility}
The classical log utility $\log(1+k)$ is differentiable in all orders in non-negative purchase quantity $k$. Moreover, the feasible region is bounded, for the consumer purchases zero product if the price exceeds $1$. 
Within the bounded region, the second-order regularity is satisfied. The third-order regularity is satisfied in the global sense. 
\subsubsection{Exponential-like Utility}
Another classical exponential utility $\left(1-e^{-k}\right)$ satisfies the conditions for a similar reason.

\section{One-Time Learning Algorithm }
Now we move to the analysis of the Algorithms using the regret tools developed above. In this section, we use $k_i(p_j)$ to denote the quantity of the purchase of item $j$ with price $p_j$ by consumers of type $i$, and we use random variable $X_t$ to denote the type of consumer at time $t$.

	\begin{algorithm}[H]
		\caption{One-Time Learning Algorithm}\label{alg:1}
		\begin{algorithmic}[1]
			\State Input: $(\{c_j\}, T)$.
			\State Initiate the price equals the cost: $p=c$.
		    \For  {$t\leq (T)^{\frac{2}{3}}(\log T)^{\frac{1}{3}}$}
		    \State Sell at the price parameter $p$ and observe the consumer type.
				\EndFor
				
				\State Compute the empirical distribution for types of consumer $X$ up to time $(T)^{\frac{2}{3}}(\log T)^{\frac{1}{3}}$:
				$$
			\hat{P}_{(T)^{\frac{2}{3}}(\log T)^{\frac{1}{3}}}^{(*)}:=$$$$\left(\frac{1}{\floor{(T)^{\frac{2}{3}}(\log T)^{\frac{1}{3}}}} \sum_{i=1}^{\floor{(T)^{\frac{2}{3}}(\log T)^{\frac{1}{3}}}} \mathbbm{1}_{\left\{X_{i}=\text{ type 1}\right\}},\cdots , \frac{1}{\floor{(T)^{\frac{2}{3}}(\log T)^{\frac{1}{3}}}} \sum_{i=1}^{\floor{(T)^{\frac{2}{3}}(\log T)^{\frac{1}{3}}}} \mathbbm{1}_{\left\{X_{i}=\text{ type n}\right\}}\right)
				$$
			\State Initiate $h=0$, and $p_0=p$.
				\For  {$(T)^{\frac{2}{3}}(\log T)^{\frac{1}{3}}\leq t\leq 2(T)^{\frac{2}{3}}(\log T)^{\frac{1}{3}}$} 
				\State Perform binary search for each type: set the test price $$\underline{p_{j,h}}=p_{j,h}-\frac{|p_{j,h-1}|}{2}$$ and $$\bar{p_{j,h}}=p_{j,h}+\frac{|p_{j,h-1}|}{2},$$ 
				where $\underline{p_{j,-1}},\bar{p_{j,-1}}$ are the lower and upper bound for the feasible $p_j$. After observing all $$\{k_1(\underline{p_{j,h}}),\cdots, k_n(\underline{p_{j,h}})\}\cup\{k_1(\bar{p_{j,h}}),\cdots, k_n(\bar{p_{j,h}})\}, $$ let
				$$p_{j,h}=\argmax_{p_{j,h}\in \{ \underline{p_{j,h}},\bar{p_{j,h}}\}} \mathbb{E}_{i\sim 	\hat{P}_L^{(*)}} k_i(p_j)\cdot (p_j-c_j).$$
				\State $h=h+1$.
				\EndFor
				\For  {$t\geq 2(T)^{\frac{2}{3}}(\log T)^{\frac{1}{3}}$} 
				\State Sell according to $p_{H((T)^{\frac{2}{3}}(\log T)^{\frac{1}{3}})}$, where $H((T)^{\frac{2}{3}}(\log T)^{\frac{1}{3}})$ is a random variable that denotes the number of successful updated price.
			
				\EndFor
			
		\end{algorithmic}
	\end{algorithm}
	Before we proceed to the discussion of the regret, let us note the complexity of this algorithm is extremely low, for this primal-based algorithm requires no computation. It first simply observes the data in a memory-free fashion to learn the distribution. Then, it perturbs the price and again observes the reactions, which requires memory storage of $2n\cdot m$ units. Hence, the overall complexity is simply $O((T)^{\frac{2}{3}}(\log T)^{\frac{1}{3}})$. Having discussed the efficiency of the algorithm, we now move to discuss its performance. 
	
	\begin{theorem}
	With the policy $\pi_1$ specified in Algorithm \ref{alg:1}, we have 
	$$R_T(\pi_1)\leq O((T)^{\frac{2}{3}}(\log T)^{\frac{1}{3}}).$$
	\end{theorem}
\begin{proof}
By Theorem \ref{regret theorem} we have 
\begin{align}\label{4.1 eq 1}
    R_T(\pi_1)&\leq 2(T)^{\frac{2}{3}}(\log T)^{\frac{1}{3}}\max_{p\in \mathbb{R}^m}\sum_{j=1}^{m}\mathbb{E}_{i\sim \mathcal{P}}\left\{\left[\frac{\partial }{\partial k_j}v_i(j,k_j)\right]^{-1}(p_j)\cdot (p_j-c_j)\right\}\\
    +& (T-2(T)^{\frac{2}{3}}(\log T)^{\frac{1}{3}})  \mathbb{E}_{i\sim \mathcal{P}}\sum_{t=2(T)^{\frac{2}{3}}(\log T)^{\frac{1}{3}}}^T\left(M\left\|p((T)^{\frac{2}{3}}(\log T)^{\frac{1}{3}})-p((T)^{\frac{2}{3}}(\log T)^{\frac{1}{3}})^*\right\|\right)\\
    +& (T-2(T)^{\frac{2}{3}}(\log T)^{\frac{1}{3}})  \mathbb{E}_{i\sim \mathcal{P}}\sum_{t=2(T)^{\frac{2}{3}}(\log T)^{\frac{1}{3}}}^TM^{\prime} d_{T V}\left(\mathcal{P}, \hat{P}_{(T)^{\frac{2}{3}}(\log T)^{\frac{1}{3}}}^{(*)}\right)
\end{align}
To prove the theorem, it suffices to show the second and third terms are $O((T)^{\frac{2}{3}}(\log T)^{\frac{1}{3}})$. We first show that 
$$T\mathbb{E}_{i\sim \mathcal{P}}d_{T V}\left(\mathcal{P}, \hat{P}_{(T)^{\frac{2}{3}}(\log T)^{\frac{1}{3}}}^{(*)}\right)\sim O((T)^{\frac{2}{3}}(\log T)^{\frac{1}{3}}).$$
Indeed, after applying the exponential convergence \cite{8}
\begin{align}
    &T\mathbb{E}_{i\sim \mathcal{P}}d_{T V}\left(\mathcal{P}, \hat{P}_{(T)^{\frac{2}{3}}(\log T)^{\frac{1}{3}}}^{(*)}\right)\\
    \leq & T\cdot \mathbb{P}\left(d_{T V}\left(\mathcal{P}, \hat{P}_{(T)^{\frac{2}{3}}(\log T)^{\frac{1}{3}}}^{(*)}\right)\geq T^{-\frac{1}{3}}\log^{\frac{1}{3}}(T) \right)+T\cdot T^{-\frac{1}{3}}\log^{\frac{1}{3}}(T)\\
    \leq &2n\exp\left( -\frac{2}{n}(T)^{\frac{2}{3}}(\log T)^{\frac{1}{3}} \left(T^{-\frac{1}{3}}\log^{\frac{1}{3}}(T)\right)^2\right)+T\cdot T^{-\frac{1}{3}}\log^{\frac{1}{3}}(T)\\
    \leq& O((T)^{\frac{2}{3}}(\log T)^{\frac{1}{3}}).
\end{align}
Now we compute the expected distance of approximated optimal pricing strategy to the true optimal pricing strategy. Since our algorithm employs a version of binary search, every update of the price shortens the distance of each $p_j(h)$ to $p_j^*$ by half. Since the first order of the revenue is bounded as discussed earlier, the convergence rate is of order $O((0.5)^h)$ where $h$ is the number of updates. Hence, we need to first provide a large deviation on the number of updates $H((T)^{\frac{2}{3}}(\log T)^{\frac{1}{3}})$. By our construction, $2(T)^{\frac{2}{3}}(\log T)^{\frac{1}{3}}$ is bounded below by a sum of Bernoulli trials with parameter $\prod_{i=1}^{n}w_i>0$. Let $$S_{(T)^{\frac{2}{3}}(\log T)^{\frac{1}{3}}}:=\sum_{i=1}^{(T)^{\frac{2}{3}}(\log T)^{\frac{1}{3}}}B(\prod_{i=1}^{n}w_i)$$
denote the sum of Bernoulli trials of this parameter. Then, by \cite{5}, we can bound the large deviation as, for some $\delta\in(0,1)$
\begin{align} &\mathbb{P}\left(S_{(T)^{\frac{2}{3}}(\log T)^{\frac{1}{3}}}\leq \delta {(T)^{\frac{2}{3}}(\log T)^{\frac{1}{3}}}\cdot \prod_{i=1}^{n}w_i \right)\\
\leq &\exp\left(-H\left(\delta\prod_{i=1}^{n}w_i,\prod_{i=1}^{n}w_i\right)(T)^{\frac{2}{3}}(\log T)^{\frac{1}{3}} \right),\end{align}
where the $H$ function is the relative entropy of the system:
$$
H(a, p) \equiv(a) \log \left(\frac{a}{p}\right)+(1-a) \log \left(\frac{1-a}{1-p}\right).
$$
With this large deviation result, we can bound the third term in \ref{4.1 eq 1}, after fixing any arbitrary $\delta$ and some large constant $m'''$:

\begin{align}
    &T\cdot \mathbb{E}_{i\sim \mathcal{P}}\left\|p\left((T)^{\frac{2}{3}}(\log T)^{\frac{1}{3}}\right)-p\left((T)^{\frac{2}{3}}(\log T)^{\frac{1}{3}}\right)^*\right\|\\
       \leq &T\cdot M'''\cdot \exp{\left(\frac{1}{2}\log 0.5\cdot \delta (T)^{\frac{2}{3}}(\log T)^{\frac{1}{3}}\cdot \prod_{i=1}^{n}w_i \right)}\\
      +&T\cdot \mathbb{P}\left(S_{(T)^{\frac{2}{3}}(\log T)^{\frac{1}{3}}}\leq \delta {(T)^{\frac{2}{3}}(\log T)^{\frac{1}{3}}}\cdot \prod_{i=1}^{n}w_i \right)\\
     \cdot &\max(\left\|p\left((T)^{\frac{2}{3}}(\log T)^{\frac{1}{3}}\right)-p\left((T)^{\frac{2}{3}}(\log T)^{\frac{1}{3}}\right)^*\right\| )\\
     \leq &O((T)^{\frac{2}{3}}(\log T)^{\frac{1}{3}}).
\end{align}
After bounding the second and third terms in \ref{4.1 eq 1}, we finally have 
$$
R_T\left(\pi_1\right) \leq O\left((T)^{\frac{2}{3}}(\log T)^{\frac{1}{3}}\right).
$$
\end{proof}
The major contribution of the regret comes from the approximation of the underlying distribution taking the form of total variation distance. The error generated in the approximation of the optimal value is actually not as large as the former. Hence, the main restriction of this type of algorithm is the incomplete information regarding the consumer population, instead of their utilities, which can be detected relatively fast by perturbing the prices; whereas the population can take a long time to survey. 

\section{Dynamical Learning Algorithm}
Unlike the one-time learning algorithm, the geometric learning algorithm updates the information while making the decision. It can be more effective since it incorporates information throughout the operation period. We continue to use the notation from the previous section, with an additional fixed quantity $L$, which is the length of the first learning period. And $S$ denotes the multiplier ratio of the geometric learning. Normally, we can take $S=2$.
	\begin{algorithm}[H]
		\caption{Dynamical Learning Algorithm}\label{alg:2}
		\begin{algorithmic}[1]
			\State Input: $(\{c_j\}, T, S, L)$.
			\State Initiate the price equals the cost: $p=c$.
			\For {$t\leq L$}
			\State Sell according to $p=c$ and compute the empirical distribution for types of consumer $X$ up to time $S$ as $\hat{P}_S^{(*)}$.
			\EndFor

		    \For  {$q\leq \log_S \frac{T}{L}$}
		    \For {$S^{q}L+\sqrt{T}\sqrt{\log T}\geq t> S^qL$}
		    
		    	\State Initiate $h=0$, and $p_0=p_{H'}$, the last updated price from the previous learning period.
			
				\State Perform binary search for each type: set the test price $$\underline{p_{j,h}}=p_{j,h}-\frac{|p_{j,h-1}|}{2}$$ and $$\bar{p_{j,h}}=p_{j,h}+\frac{|p_{j,h-1}|}{2},$$	where $\underline{p_{j,-1}},\bar{p_{j,-1}}$ are the lower and upper bound for the feasible $p_j$.	After observing all $$\{k_1(\underline{p_{j,h}}),\cdots, k_n(\underline{p_{j,h}})\}\cup\{k_1(\bar{p_{j,h}}),\cdots, k_n(\bar{p_{j,h}})\} $$ let
				$$p_{j,h}=\argmax_{p_{j,h}\in \{ \underline{p_{j,h}},\bar{p_{j,h}}\}} \mathbb{E}_{i\sim 	\hat{P}_L^{(*)}} k_i(p_j)\cdot (p_j-c_j).$$
				\State $h=h+1$.
				\EndFor
				
				\For {$S^{q+1}L\geq t> S^qL+\sqrt{T}\sqrt{\log T}$}
				
			\State Sell according to $p_{(q,H(\sqrt{T}\sqrt{\log T}))}$ where ${(q,H(\sqrt{T}\sqrt{\log T}))}$ is the index of the last updated price. 
			\EndFor
			\State Compute the empirical distribution $\hat{P}_{S^{q+1}L}^{(*)}$.
		
				\EndFor
			
		\end{algorithmic}
	\end{algorithm}
	Essentially what this algorithm does is to accumulate the information regarding the consumer distribution throughout the learning periods, and make necessary changes to the pricing strategy given better approximated distribution. Like the previous algorithm, it only requires memory storage of $2n\cdot m$ units for the price perturbation, while the learning of the distribution is memory-free. The overall complexity of this algorithm is simply $O(T)$, as the learning is spread over the entire period. This algorithm, in addition, takes a similar form as \cite{28} which gives an order of regret $O(\sqrt{T} \log(T)^{4.5})$ under a different context.

		\begin{theorem}\label{thm alg2}
	With the policy $\pi_2$ specified in Algorithm \ref{alg:2}, 
	$$R_T(\pi_1)\leq O((T)^{\frac{1}{2}}(\log T)^{\frac{1}{2}})$$
	\end{theorem}
	\begin{proof}
	First, we claim that it suffices to show for any $k$ whose learning period is $k/S$, we have, for some large constant $B$, we have 
	\begin{align}
	    f(k):=&\sum_{k\geq t\geq k/S} \sum_{j=1}^{m}\mathbb{E}_{i \sim \mathcal{P}}\left\{
\left[\frac{\partial}{\partial k_j} v_i\left(j, k_j\right)\right]^{-1}\left(p_j^*\right) \cdot\left(p_j^*-c_j\right)\right\}\\
-&\sum_{k\geq t\geq k/S} \sum_{j=1}^{m}\mathbb{E}_{i \sim \mathcal{P}}\left\{\left[\frac{\partial}{\partial k_j} v_i\left(j, k_j\right)\right]^{-1}\left(p_j(t)\right) \cdot\left(p_j(t)-c_j\right)\right\}\\
\leq &B(\sqrt{k}\sqrt{\log k}).
	\end{align}
	Indeed, suppose the above if true, then
		\begin{align}
&	R_T(\pi_2)=\sum_{l\leq \log_S(T/L)+1} \mathbb{E} \sum_{t=\floor{(S^{l-1})L+1}}^{\floor{(S^{l})L}} 
\sum_{j=1}^m \mathbb{E}_{i \sim \mathcal{P}}\\
&\left\{
\left[\frac{\partial}{\partial k_j} v_i\left(j, k_j\right)\right]^{-1}\left(p_j^*\right) \cdot\left(p_j^*-c_j\right)-
\left[\frac{\partial}{\partial k_j} v_i\left(j, k_j\right)\right]^{-1}\left(p_j(t)\right) \cdot\left(p_j(t)-c_j\right)\right\}
\\
	\leq & B \sum_{l\leq \log_S(T/L)+1} (\sqrt{LS^l}\sqrt{\log LS^l})\\
	\leq & B \sqrt{T}\sqrt{\log T}+ B\sum_{h=0}^{\infty} \sqrt{\frac{T}{S^h}}\sqrt{\log \frac{T}{S^h}}\\
	\leq & B \sqrt{T}\sqrt{\log T}+2B\sqrt{T}\sqrt{\log T} \left(\sum_{h=0}^{\infty} \sqrt{\frac{1}{S^h}}\right)\\
	\leq & O(\sqrt{T}\sqrt{\log T}).
\end{align}
	Now, we will prove the bound required for the claim. Indeed, similar to what we did before,
	\begin{align}
	    f(k)\leq & \mathbb{E}_{i \sim \mathcal{P}}\sum_{k\geq t\geq k/S} \left( M\| p(k+\sqrt{k}\sqrt{\log k})-p(k+\sqrt{k}\sqrt{\log k})^*\|\right)\\
	    +& \mathbb{E}_{i \sim \mathcal{P}}\sum_{k\geq t\geq k/S} M^{\prime} d_{T V}\left(\mathcal{P}, \hat{P}_{k/S}^{*}\right)\\
	    \leq &k\cdot M'''\cdot \exp{\left(\frac{1}{2}\log 0.5\cdot \delta (k)^{\frac{1}{2}}(\log k)^{\frac{1}{2}}\cdot \prod_{i=1}^{n}w_i \right)}\\
      +&k\cdot \mathbb{P}\left(S_{\sqrt{k}\sqrt{\log k}}\leq \delta \sqrt{k}\sqrt{\log k}\cdot \prod_{i=1}^{n}w_i \right)\\
     \cdot &\max(\left\|p\left((k)^{\frac{1}{2}}(\log k)^{-\frac{1}{2}}\right)-p\left((k/S)^{\frac{1}{2}}(\log k/S)^{-\frac{1}{2}}\right)^*\right\| )\\
     +&k\cdot \mathbb{P}\left(d_{T V}\left(\mathcal{P}, \hat{P}_{k/S}^{(*)}\right)\geq \sqrt{\frac{LS}{2}} k^{-\frac{1}{2}}\log^{\frac{1}{2}}(k) \right)+k\cdot k^{-\frac{1}{2}}\log^{\frac{1}{2}}(k)\\
    \leq &O(\sqrt{k}\sqrt{\log k})
	\end{align}
	where we have $$k\cdot M'''\cdot \exp{\left(\frac{1}{2}\log 0.5\cdot \delta (k)^{\frac{1}{2}}(\log k)^{\frac{1}{2}}\cdot \prod_{i=1}^{n}w_i \right)}\leq O(\sqrt{k}),$$
	then,

	\begin{align}&k\cdot \mathbb{P}\left(S_{\sqrt{k}\sqrt{\log k}}\leq \delta \sqrt{k}\sqrt{\log k}\cdot \prod_{i=1}^{n}w_i \right)\\
	\leq &k\cdot \exp\left(-H\left(\prod_{i=1}^{n}w_i, \delta \prod_{i=1}^{n}w_i\right)\sqrt{k}\sqrt{\log k}\right)\\
	\leq & O(\sqrt{k}),
	\end{align}
	
	and finally,
	
	\begin{align}
	   & k\cdot \mathbb{P}\left(d_{T V}\left(\mathcal{P}, \hat{P}_{k/S}^{(*)}\right)\geq \sqrt{\frac{LS}{2}} k^{-\frac{1}{2}}\log^{\frac{1}{2}}(k) \right)\\
	  \leq &2 k \exp \left(-\frac{2 k}{S n}\left(  \sqrt{\frac{Sn}{2}} k^{-\frac{1}{2}} \log ^{\frac{1}{2}}(k)\right)^2\right)\\
	  \leq &O(\sqrt{k}).
	\end{align}
	So the proof is complete.
	
	\end{proof}

\section{Lower Bounds of Regrets}
	In this section, we will provide a simple example to illustrate a lower bound for this type of problem. Let us consider two types of consumers and one type of product at a cost 0.1 unit. Type I consumer represents 10 percent of the population with utility $\log(1+k)$ while type II has utility $1-\exp(-k)$, where $k$ is the quantity purchased, both satisfy the regularity conditions discussed in Theorem \ref{regret theorem}. The optimization problem now becomes 
	
	\begin{align}
	   & \max_{p\in \mathbb{R}^m}\sum_{j=1}^{m}\mathbb{E}_{i\sim \mathcal{P}}\left\{\left[\frac{\partial }{\partial k_j}v_i(j,k_j)\right]^{-1}(p_j)\cdot (p_j-c_j)\right\}\\
	   =& \left(w_1\left(\frac{1}{p}-1\right)+\left(1-w_1\right)\left(-\ln p\right)\right)\left(p-0.1\right)
	\end{align}
where $w_1=0.1$ and the maximization occurs at $p=0.438$ with optimal expected revenue of $0.292$. If we have the empirical distribution showing $\tilde{w}_1\geq 0.9$, then we have $p\leq 0.324$ such that the revenue collected is less than $0.274$, a ten percent decrease. Hence, for any one-time learning algorithm $\pi_1^*$ of learning length $t$, we have 

\begin{align}
   & R_T(\pi_1^*)\geq O(t)+O\left((T-t)\mathbb{P}\left(\sum_{i=1}^{t}B(w_1)\geq (1-w_1)t \right)\right)\\
     \geq& O(t)+O\left(\frac{T-t}{\sqrt{2 t}} \exp \left(-t D\left(1-w_1 \| w_1\right)\right)\right)\\
     \geq &O(\sqrt{\log T})
\end{align}
where the second to last inequality uses a lower bound on Bernoulli sample sum concentration \cite{26} using relative entropy $D(\cdot \| \cdot)$ and the last line comes from optimizing the choice of $t$ as a function of $T$. In particular, there is a trade-off between exploration and exploitation, and the inequality essentially shows we cannot spend less than $O(\sqrt{\log T})$ to explore. Hence, there already exists one example whose regret cannot be lower than the order of $O(\sqrt{\log T})$:

\begin{theorem} \label{lower 1}
    Under the conditions of Theorem \ref{regret theorem}, there exists a model such that any one-time learning algorithm $\pi_1^*$ has a regret bounded below by 
   $$ R_T(\pi_1^*)\geq O(\sqrt{\log T}).$$
\end{theorem}
	
	Now, we use the same model yet the utilities are no longer known to the seller, we can show for any geometric update algorithm $\pi_2^*$ with initial learning period $L$ and price perturbation period of $t$. Suppose we never have less than two successful updates, then surely the algorithm will fail to learn the price using binary search. The probability of this no-learning event in a period of length $t$ is bounded below by $(1-p_j)^t$ where $p_j$ is the distribution of any type of consumer. Then, for Bernoulli distribution of $B(\cdot)$ and Geometric distribution $G(\cdot)$:
	
	\begin{align}
   & R_T(\pi_2^*)\geq \sum_{i=1}^{\floor{\log_S T/L}-1}O\left((LS^{i+1}-LS^i)\mathbb{P}\left(\sum_{i=LS^i}^{t}B(w_1)\geq (1-w_1)t \right)\right)\\
   +& \sum_{i=1}^{\floor{\log_S T/L}-1}O\left[\left(LS^{i+1}-LS^i\right)\mathbb{P}\left( G(p_j)\geq t\right)\right]\\
    \geq& \sum_{i=1}^{\floor{\log_S T/L}} O\left(\frac{LS^{i+1}-LS^i}{\sqrt{2 N^i}} \exp \left(-S^iL D\left(1-w_1 \| w_1\right)\right)\right)\\
     +&\sum_{i=1}^{\floor{\log_S T/L}}O\left[\left(LS^{i+1}-LS^i\right)(1-p_j)^tp_j \right]\\
     \geq &O(\sqrt{\log T})
\end{align}
	The last inequality holds because again, by the exploration and exploitation, learning less than $O(\sqrt{\log T})$ would cause the binary search for the optimal price to under-perform, let alone additional errors from the distribution approximation. This leads to the following theorem:
	\begin{theorem}\label{lower 2}
    Under the conditions of Theorem \ref{regret theorem}, there exists a model such that Algorithm \ref{alg:2} $\pi_2$ has a regret bounded below by 
   $$ R_T(\pi_2)\geq O(\sqrt{\log T}).$$
\end{theorem}
Let us be careful that while the previous theorem holds for any one-time learning algorithm, this theorem holds only for our proposed geometric update algorithm. We believe that better geometric algorithms exist under different contexts, so we cannot find one example that undermines all geometric learning algorithms. Combing Theorem \ref{thm alg2}, Theorem \ref{lower 1} and Theorem \ref{lower 2}, we know the lower and upper bounds leave a region of order $O(\sqrt{T})$, where future improvements can be made. This lower bound result is different from \cite{28}, which shows the lower bound can also be made to the order $O(\sqrt{T})$. The example they use is called "uninformative price-learning". However, this is due to the fact that they require a more explicit search on the consumer utility, while our binary search does not, as long as the third-order smoothness condition is imposed. Hence, every step of our search is equally "informative". As a result, their example that establishes $O(\sqrt{T})$ lower bound cannot be used in our context. 
	
		\section*{Acknowledgement}
I am grateful for the continuous support from Yinyu Ye, without whom this paper cannot be made possible. Conversations with Devansh Jalota are also helpful.

	\bibliographystyle{apalike}
	
	\bibliography{OF}
	
\end{document}